\documentclass[11pt]{article}

\usepackage[margin=1.1in]{geometry}
\usepackage{amsmath,amssymb,amsthm}
\usepackage{mathtools}
\usepackage[colorlinks=true,linkcolor=blue,citecolor=blue,urlcolor=blue]{hyperref}
\usepackage{enumitem}
\usepackage{comment}
\newtheorem{theorem}{Theorem}[section]
\newtheorem{lemma}[theorem]{Lemma}
\newtheorem{proposition}[theorem]{Proposition}
\newtheorem{corollary}[theorem]{Corollary}

\theoremstyle{definition}
\newtheorem{definition}[theorem]{Definition}
\newtheorem{example}[theorem]{Example}
\theoremstyle{remark}
\newtheorem{remark}[theorem]{Remark}

\newcommand{\Q}{\mathbb{Q}}
\newcommand{\Z}{\mathbb{Z}}
\newcommand{\F}{\mathbb{F}}
\newcommand{\fa}[1]{\Q\langle #1\rangle}                         
\newcommand{\ff}[1]{\Q(\!\langle #1\rangle\!)}                   
\newcommand{\ncrk}{\operatorname{nc\text{-}rank}}
\newcommand{\rk}{\operatorname{rank}}
\newcommand{\val}{\operatorname{val}}
\renewcommand{\r}{\mathfrak{r}}
\newcommand{\NC}{\mathsf{NC}}
\newcommand{\AC}{\mathsf{AC}}
\newcommand{\RNC}{\mathsf{RNC}}
\newcommand{\PSPACE}{\mathsf{PSPACE}}
\newcommand{\EXPSPACE}{\mathsf{EXPSPACE}}
\newcommand{\DSPACE}{\mathsf{DSPACE}}
\newcommand{\BPSPACE}{\mathsf{BPSPACE}}
\newcommand{\ncFULL}{\textup{\textsc{ncFull}}}

\title{Rational Identity Testing for Noncommutative Circuits\\ is in Polynomial Space}
\author{V.\ Arvind\thanks{Institute of Mathematical Sciences, Chennai, India. \texttt{arvind@imsc.res.in}}
\and Pushkar S.\ Joglekar\thanks{Vishwakarma Institute of Technology, Pune, India. \texttt{joglekar.pushkar@gmail.com}}}
\date{September 27, 2026}

\begin{document}
\maketitle

\begin{abstract}
  The rational identity testing problem, RIT, asks whether an input \emph{rational circuit}, i.e.\ a noncommutative arithmetic circuit that has inversion gates, computes the zero element of the free skew field. For rational \emph{formulas} the problem is known to be in deterministic polynomial time \cite{GGOW20}, \cite{IQS18}. For rational circuits of size $s$ the associated linear pencil has dimension $2^{O(s)}$, and the known algorithms use exponential time and exponential space.

  We show in this note that RIT for rational circuits over $\Q$ and finite fields, including the question whether the circuit is a correct computation at every inversion gate, is decidable in deterministic space $O(s^{O(1)})$, and hence is in $\PSPACE$. As a consequence, noncommutative PIT for degree unrestricted input noncommutative circuits is also in $\PSPACE$. The proof depends on algorithms that take succinct representations of their input.
  It has the following three observations that we put together to obtain the
  $\PSPACE$ upper bound.
\begin{itemize}  
\item 
Analyzing the Hrubes-Wigderson construction of a linear pencil for an input rational formula, we obtain a succinctly represented linear pencil $A$ of size $2^{O(s)}$ for the input rational \emph{circuit} of size $s$. More precisely, given indices $i$ and $j$ for the linear pencil we can compute $A_{i,j}$ in space polynomial in $s$ and length of $i,j$. This algorithm essentially gives a succinct presentation of the linear pencil $A$ obtained by the reduction.
\item
By the recent theorem of Chatterjee, Ghosh, Gurjar, Raj and Thierauf that deciding whether or not a symbolic matrix $\sum_i A_ix_i$ has full noncommutative rank is in $\NC$. We note that their algorithm is
  logspace-uniform $\NC$ and we use it as a black-box for computing rank of succinctly represented pencil $A$.
 
\item  
Finally, we note a folklore simulation which we state and prove in the form we need. If a problem is solved by a logspace-uniform family of \emph{deterministic} Boolean circuits of polylogarithmic depth, and its input is not written down but is presented by a polylogspace subroutine that returns any requested input bit, then the circuit can be evaluated in polylogarithmic space. Hence, simulating a depth $O(\log^{i}M)$ circuit for succinctly presented inputs of length $M = 2^{\Theta(s)}$ yields a polynomial space algorithm. 
\end{itemize}

\end{abstract}


\section{Introduction}\label{sec:intro}

Let $\fa{X} = \fa{x_1,\dots,x_n}$ be the free associative algebra over $\Q$ in noncommuting variables and let $\ff{X}$ be its universal field of fractions, the free skew field of Amitsur and Cohn \cite{Coh95,Coh06}. A \emph{noncommutative rational circuit} is an arithmetic circuit with gates $\{+,\times,{}^{-1}\}$ over $\Q$ and inputs from $X \cup \Q$. An inversion gate is legal precisely when its argument is a nonzero element of the skew field. A circuit all of whose inversion gates are legal is called \emph{correct}, and it computes an element $\widehat{\Phi}$ of $\ff{X}$. The rational identity testing problem, RIT, is to decide, given a circuit $\Phi$, whether $\Phi$ is correct, and if it is, whether $\widehat{\Phi} = 0$.

\paragraph{Prior work.}
The algorithmic study of RIT was initiated by Hrube\v{s} and Wigderson \cite{HW15}. They reduced RIT for rational \emph{formulas} to the question whether an associated \emph{linear pencil}
\[
L \;=\; Q_0 + \sum_{k=1}^{n} Q_k x_k , \qquad Q_k \in \Q^{N \times N},
\]
of dimension $N$ linear in the formula size, is \emph{full}, that is, invertible over the free skew field. Through this reduction, RIT for formulas is known to be in $\PSPACE$ by the Cohn--Reutenauer criterion \cite{CR99}, which expresses non-fullness of a polynomial-size pencil as the feasibility of a system of commutative polynomial equations in polynomially many unknowns. It is in deterministic polynomial time in the white-box model, by the operator scaling algorithm of Garg, Gurvits, Oliveira and Wigderson \cite{GGOW16,GGOW20} over $\Q$, and by the constructive rank algorithm of Ivanyos, Qiao and Subrahmanyam \cite{IQS18} over $\Q$ and also over finite fields. We also know a deterministic $\NC$-Turing reduction from RIT for formulas to the linear matrix rank problem \cite{AJ24}. In the black-box model there is a deterministic quasipolynomial-time algorithm \cite{ACM24stoc}. 
For rational \emph{circuits} the situation is very different. A circuit of size $s$ unrolls into a formula of size $2^{O(s)}$, so its pencil has dimension $2^{O(s)}$. The polynomial-time rank algorithms run on that pencil in time and space $2^{O(s)}$, and the Cohn--Reutenauer route meets a polynomial system with $2^{O(s)}$ unknowns, giving only $\EXPSPACE$. To the best of our knowledge, no polynomial-space upper bound for RIT on circuits was known.

\paragraph{Proof Overview.}
\begin{enumerate}
\item The first step is some well-known complexity theory. A deterministic (logspace) uniform family ofbounded fanin Boolean circuits of depth $D(M)$ can be evaluated in deterministic workspace $O(D(M))$, provided $D(M)$ is at least logarithmic in the circuit size \cite{Bor77,Ruz81}. For uniform $\NC^{i}$ this is the containment $\NC^{i} \subseteq \DSPACE(\log^{i} M)$. It is known that presenting an input \emph{succinctly}, by a small Boolean circuit that returns the bit at a given position, raises the complexity of a problem by one exponential \cite{GW83,PY86,BLT92} given that the succinct presentation
circuit of size $s$ is for inputs of length $2^{O(s)}$.

\item Putting the two together, suppose the input string has length $M = 2^{\Theta(s)}$ but is presented by a succinct circuit description of size $s$, then a $\DSPACE(\log^{i}M)$ algorithm essentially becomes a $\DSPACE(s^{i})$ algorithm, because $\log^{i}(2^{\Theta(s)}) = \Theta(s^{i})$. A space-bounded machine never needs its input written down. It holds a position index, which is only $\Theta(s)$ bits, and calls the succinct description circuit whenever it wants a bit. In other words, the succinct version of any problem in deterministic uniform $\NC$ is in $\PSPACE$. This is folklore. Nevertheless, we state it as Theorem~\ref{thm:succsim} and prove it, because the exact hypotheses matter for the main result.

\item Now, coming to the RIT problem for rational circuits, based on Hrubes-Wigderson construction \cite{HW15}, we observe that a rational circuit of size $s$ linearizes into a pencil $L$ of dimension $N = 2^{O(s)}$ whose fullness of rank decides the identity testing question. Although the pencil $L$ is too large to write down, we argue that each entry $L[i,j]$ is computable from the indices $i,j$ in space $O(s^{O(1)})$. This algorithm therefore is a succinct presentation of the linear pencil $L$ as input to the noncommutative rank problem, where the succinct representation's size is $O(s^{O(1)})$. 

\item Now, we will crucially need the result of Chatterjee, Ghosh, Gurjar, Raj and Thierauf \cite{CGGRT26} who proved checking if a homogeneous noncommutative linear pencil $\sum_i A_i x_i$ has full noncommutative rank is in $\NC$. We will apply the polylogspace simulation, described above, of the logspace-uniform $\NC^{2}$ algorithm of \cite{CGGRT26} with input as the succinctly presented linear pencil $L$. This entire simulation can be carried out in deterministic space $O(\log^{2}M) + O(s^{O(1)}) = O(s^{O(1)})$ which is essentially
polynomial in $s$.
\end{enumerate}


\begin{theorem}[Theorem~\ref{thm:main}]

  There is a deterministic Turing machine which, given a noncommutative rational circuit $\Phi$ of size $s$ over $\Q$ in variables $x_1,\dots,x_n$, decides in space $O(s^{O(1)})$:
\begin{enumerate}[label=(\alph*),nosep]
\item whether $\Phi$ is a correct circuit, identifying an inversion gate whose argument evaluates to $0$ if it is not; and
\item if $\Phi$ is correct, whether $\widehat{\Phi} = 0$ in the free skew field.
\end{enumerate}
In particular noncommutative PIT for degree-unrestricted noncommutative circuits is in $\PSPACE$.
\end{theorem}

\section{Uniform parallel algorithms on succinct inputs}\label{sec:sim}

Throughout this section ``workspace'' counts worktape cells of a deterministic Turing machine. A region of workspace is called \emph{reusable} if it may be overwritten between calls, so that its size is added and not multiplied. Logarithms are base $2$.

\subsection{Succinct inputs}

\begin{definition}[succinct presentation; bit oracle]\label{def:succinct}
Let $M \in \mathbb{N}$ and $x \in \{0,1\}^{M}$. A \emph{succinct presentation of $x$ with parameters $(b,\sigma)$} consists of
\begin{enumerate}[label=(\roman*),nosep]
\item the number $M$, written in binary, with $M \le 2^{b}$; and
\item a deterministic subroutine, the \emph{bit oracle}, which on input $j \in [M]$ written in binary outputs the bit $x_j$, using reusable workspace at most $\sigma$.
\end{enumerate}
An algorithm \emph{decides $L$ on succinctly presented inputs} if, given $b$ and access to the bit oracle as a subroutine, it decides whether $x \in L$.
\end{definition}

The case described in the introduction is the special case in which the bit oracle is a Boolean circuit.

\begin{example}[the circuit presentation]\label{ex:circuitpres}
Let $C$ be a Boolean circuit with $n$ inputs, one output and $s$ gates, and let
\[
x_C \;:=\; \bigl(C(\langle 1\rangle),\,C(\langle 2 \rangle),\,\dots,\,C(\langle 2^{n}\rangle)\bigr) \;\in\; \{0,1\}^{2^{n}},
\]
where $\langle j \rangle$ is the $n$-bit binary encoding of $j-1$. Then $C$ is a succinct presentation of $x_C$ with $M = 2^{n}$, $b = n$ and $\sigma = O(s)$. Indeed, evaluating $C$ at a given input is an instance of the circuit value problem, and a circuit of size $s$ is evaluated in workspace $O(s)$ by computing the gate values in topological order and storing all $s$ of them. Here $n \le s$, so $b \le s$.
\end{example}

We will also use a presentation which is not a Boolean circuit. In Section~\ref{sec:main} the bit oracle is the entry oracle of a linear pencil. Definition~\ref{def:succinct} is deliberately agnostic about how the oracle is implemented, and only its workspace $\sigma$ enters the bounds.

\subsection{Uniform circuit families}

\begin{definition}[uniform family]\label{def:uniform}
Let $\mathcal{C} = \{C_M\}_{M \ge 1}$ be a family of Boolean circuits where $C_M$ has $M$ inputs, one output and fan-in at most $2$. The family has \emph{size} $S(M)$ and \emph{depth} $D(M)$ if $C_M$ has at most $S(M)$ gates and depth at most $D(M)$. Gates of $C_M$ are named by binary strings of length $\lambda(M) = O(\log S(M))$.

The family is \emph{uniform with uniformity space $u(M)$} if there is a single deterministic procedure which, given $M$ in binary and using reusable workspace at most $u(M)$,
\begin{enumerate}[label=(\roman*),nosep]
\item returns the name of the output gate of $C_M$;
\item on a gate name $g$, returns the \emph{type} of $g$, which is one of $\wedge$, $\vee$, $\neg$, a Boolean constant, or ``input gate number $j$'' together with $j \in [M]$ in binary; and
\item on a gate name $g$ and a child index $c \in \{1,2\}$, returns the name of the $c$-th child of $g$.
\end{enumerate}
\end{definition}

Items (i) to (iii) are the usual \emph{direct connection} presentation of a circuit family. Both of the standard uniformity conditions supply them within the space we need. $\mathsf{DLOGTIME}$-uniformity in the sense of Ruzzo \cite{Ruz81}, which is the default for $\NC$, decides the direct connection language in time $O(\log M)$ on a random-access machine and hence in space $O(\log M)$. Logspace-uniformity gives $u(M) = O(\log M)$ as well.\footnote{Logspace-uniformity is usually defined by a transducer that maps $1^{M}$ to a description of $C_M$ in space $O(\log M)$. Such a transducer is simulated without $1^{M}$ physically present: store its input-head position, an $O(\log M)$-bit counter, and answer every read with $1$ except at the two ends. Its output is a description of $C_M$ and has length $\mathrm{poly}(M)$, so it too is not stored, but is scanned as it is produced on the write-only output tape. The simulator retains only the queried gate name, the field it is currently reading, and a comparison pointer, which is $O(\log M)$ bits in all.}

\begin{definition}[uniform $\NC^{i}$]\label{def:nci}
For $i \ge 1$, a language $L$ is in uniform $\NC^{i}$ if it is decided by a uniform family as in Definition~\ref{def:uniform} with $S(M) = M^{O(1)}$, $D(M) = O(\log^{i}M)$ and $u(M) = O(\log M)$.
\end{definition}

\subsection{The simulation}

\begin{theorem}[succinct simulation]\label{thm:succsim}
Let $L$ be decided by a uniform family $\{C_M\}$ of Boolean circuits of fan-in at most $2$, of size $S(M)$, depth $D(M)$ and uniformity space $u(M)$. Suppose $x \in \{0,1\}^{M}$ is given by a succinct presentation with parameters $(b,\sigma)$. Then whether $x \in L$ is decidable in deterministic workspace
\[
O\bigl(D(M) \;+\; \log S(M) \;+\; b \;+\; u(M)\bigr) \;+\; \sigma ,
\]
and in time $2^{O(D(M))}\cdot \mathrm{poly}\bigl(D(M) + \log S(M) + b\bigr)\cdot T_{\mathrm{or}}$, where $T_{\mathrm{or}}$ bounds the running time of the bit oracle and of the uniformity procedure.
\end{theorem}

\begin{proof}
The algorithm is a depth-first evaluation of $C_M$ which stores no gate values except along the current root-path, and which stores no gate names at all.

\medskip\noindent\emph{Data.} We maintain a \emph{path register} $\pi \in \{1,2\}^{\le D(M)}$, which is the sequence of child-choices taken from the output gate down to the gate currently under evaluation, and alongside each symbol of $\pi$ one bit $\mathrm{st}$ recording the value of the first child at that level once that value has been computed. We also maintain a reusable \emph{scratch region} of $O(\lambda(M) + b + u(M)) = O(\log S(M) + b + u(M))$ bits holding one gate name, one input index, and the workspace of the uniformity procedure, together with the reusable oracle region of $\sigma$ bits. The path register and its attached bits occupy $2D(M)$ bits. No further storage is used, and in particular no recursion stack of gate names is kept.

\medskip\noindent\emph{Recovering the current gate.} Define a subroutine $\mathrm{NAME}(\pi)$ as follows. Obtain the name of the output gate from Definition~\ref{def:uniform}(i), and then for $t = 1,\dots,|\pi|$ replace the name currently held by the name of its $\pi_t$-th child, using Definition~\ref{def:uniform}(iii). This is a loop and not a recursion: only one gate name is ever held, and the workspace of the uniformity procedure is reused at every step. Its cost is $O(\lambda(M) + u(M))$ workspace and $|\pi| \le D(M)$ calls. This is the step at which uniformity is used, and it is what allows the path to be stored as $O(1)$ bits per level rather than as a list of gate names.

\medskip\noindent\emph{Main loop.} Start with $\pi$ empty, in the descending phase.

\emph{Descend.} Compute $g := \mathrm{NAME}(\pi)$ and its type.
\begin{itemize}[nosep]
\item If $g$ is a Boolean constant, set $\mathrm{val}$ to that constant and ascend.
\item If $g$ is the input gate number $j$, call the bit oracle on $j$, set $\mathrm{val} := x_j$, and ascend. The index $j$ has at most $b$ bits and is held in the scratch region.
\item Otherwise $g$ is a $\neg$, $\wedge$ or $\vee$ gate. Append the symbol $1$ to $\pi$, with its attached bit unset, and descend again.
\end{itemize}

\emph{Ascend} with the value $\mathrm{val}$ of the gate at the end of $\pi$.
\begin{itemize}[nosep]
\item If $\pi$ is empty, output $\mathrm{val}$ and halt.
\item Otherwise let $c$ be the last symbol of $\pi$ and $\mathrm{st}$ its attached bit. Delete that symbol and recompute $g := \mathrm{NAME}(\pi)$ and its type, so that $g$ is the parent. If $g$ is $\neg$, replace $\mathrm{val}$ by $\neg\,\mathrm{val}$ and ascend. If $g$ is $\wedge$ or $\vee$ and $c = 1$, append the symbol $2$ to $\pi$ with attached bit $\mathrm{val}$, and descend. If $g$ is $\wedge$ or $\vee$ and $c = 2$, replace $\mathrm{val}$ by $\mathrm{st}\wedge\mathrm{val}$ or $\mathrm{st}\vee\mathrm{val}$ as appropriate, and ascend.
\end{itemize}

\medskip\noindent\emph{Correctness.} By induction on $D(M) - |\pi|$, the procedure ascends out of a path $\pi$ carrying in $\mathrm{val}$ exactly the value that the gate at the end of $\pi$ takes when $C_M$ is evaluated on $x$. The base cases are constants and input gates, and the inductive step is the definition of $\neg$, $\wedge$ and $\vee$. This is the only place where determinism is used, and it is used silently: a gate with several root-paths leading to it is evaluated once for each of them, and the induction is valid because all of these evaluations return the same value. The output is $C_M(x)$, which is $1$ if and only if $x \in L$.

\medskip\noindent\emph{Resources.} The workspace is $2D(M)$ for the path register, $O(\log S(M) + b + u(M))$ for the scratch region, and $\sigma$ for the oracle region. The last two are reused rather than accumulated, which gives the stated bound. For the running time, note that every configuration of the algorithm is determined by $\pi$ together with the attached bits and a phase, so the number of steps is at most $2^{O(D(M))}$, and each step costs one $\mathrm{NAME}$ computation, that is $O(D(M))$ calls to the uniformity procedure, plus at most one oracle call.
\end{proof}

\begin{corollary}[succinct $\NC^{i}$ is in $\PSPACE$]\label{cor:succNC}
Let $i \ge 1$ and let $L$ be in uniform $\NC^{i}$. Suppose $x \in \{0,1\}^{M}$ is given by a succinct presentation with parameters $(b,\sigma)$ where $b = O(s)$, so that $M \le 2^{O(s)}$. Then whether $x \in L$ is decidable in deterministic workspace
\[
O(s^{i}) \;+\; \sigma
\]
and in time $2^{O(s^{i})}\cdot\mathrm{poly}(s)\cdot T_{\mathrm{or}}$. In particular, if $\sigma = \mathrm{poly}(s)$ then the problem is in $\PSPACE$ measured in $s$.
\end{corollary}

\begin{proof}
Apply Theorem~\ref{thm:succsim}. With $M \le 2^{O(s)}$ we get $D(M) = O(\log^{i}M) = O(s^{i})$, $\log S(M) = O(\log M) = O(s)$ and $u(M) = O(\log M) = O(s)$. Since $i \ge 1$ the sum is $O(s^{i}) + \sigma$.
\end{proof}

\begin{corollary}[the circuit presentation]\label{cor:circpres}
Let $i \ge 1$, let $L$ be in uniform $\NC^{i}$, and let
\[
\mathrm{Succinct}(L) \;:=\; \bigl\{\,C \;:\; C \text{ a Boolean circuit with } n \text{ inputs and one output, and } x_C \in L \,\bigr\}
\]
with $x_C \in \{0,1\}^{2^{n}}$ as in Example~\ref{ex:circuitpres}. Then $\mathrm{Succinct}(L) \in \DSPACE(O(s^{i}))$, where $s$ is the size of $C$. In particular $\mathrm{Succinct}(L) \in \PSPACE$.
\end{corollary}

\begin{proof}
Example~\ref{ex:circuitpres} gives $b = n \le s$ and $\sigma = O(s)$. Apply Corollary~\ref{cor:succNC}.
\end{proof}

\begin{remark}[relation to the succinct-input literature]\label{rem:lit}
That succinct presentation raises complexity by one exponential is a standard theme. It goes back to Galperin and Wigderson \cite{GW83} and to Papadimitriou and Yannakakis \cite{PY86}, and was systematised as a family of upgrade theorems by Balc\'azar, Lozano and Tor\'an \cite{BLT92}: a problem complete for a class $\mathcal{C}$ under suitably restricted reductions has a succinct version complete for the exponentially larger class. Corollary~\ref{cor:succNC} is the space-bounded instance of that theme, and we do not claim it as new. We have stated it with an explicit bit-oracle interface rather than for a fixed complete problem, because that is the form in which Section~\ref{sec:main} needs it. Our succinct object is a linear pencil presented by an entry oracle, and not a Boolean circuit.
\end{remark}

\section{Rational circuits and the free skew field}\label{sec:prelim}

In this section we recall the basic definitions and fix the notation.

\subsection{Inner rank and fullness}

Throughout, $\fa{X} = \fa{x_1,\dots,x_n}$ is the free associative $\Q$-algebra and $\ff{X}$ is the free skew field, the universal field of fractions of $\fa{X}$ \cite{Coh95,Coh06}. The canonical map $\fa{X}\hookrightarrow \ff{X}$ is injective.

\begin{definition}[inner rank; noncommutative rank]\label{def:ncrank}
Let $A$ be a matrix over $\fa{X}$, not necessarily square. The \emph{inner rank} of $A$ is the least $\rho$ for which there is a factorization $A = BC$ with $B$ having $\rho$ columns and $C$ having $\rho$ rows, both over $\fa{X}$. We write $\ncrk(A)$ for the inner rank and also call it the \emph{noncommutative rank} of $A$. A square matrix $A \in \fa{X}^{N\times N}$ is \emph{full} if $\ncrk(A) = N$.
\end{definition}

\begin{proposition}[Cohn]\label{prop:dictionary}
For $A \in \fa{x_1,\dots,x_m}^{N\times N}$ the following are equivalent: (i) $A$ is invertible over $\ff{x_1,\dots,x_m}$; (ii) $A$ has rank $N$ over $\ff{x_1,\dots,x_m}$; (iii) $A$ is full. More generally, for an arbitrary matrix $A$ over $\fa{X}$, $\ncrk(A)$ equals the rank of $A$ over $\ff{X}$.
\end{proposition}

\begin{proof}
The equivalence of (i) and (ii) is Gaussian elimination over a skew field. For the general statement, and hence for the equivalence of (ii) and (iii), recall that $\fa{X}$ is a free ideal ring and hence a Sylvester domain, and that $\ff{X}$ is its universal field of fractions. Over a Sylvester domain the inner rank of a matrix equals its rank over the universal field of fractions \cite[Ch.~5, 7]{Coh06}.
\end{proof}

An \emph{affine linear pencil} is a matrix $L = Q_0 + \sum_{k=1}^{n} Q_k x_k$ with $Q_k \in \Q^{N \times N}$. A \emph{homogeneous} pencil is one in which $Q_0$ has been absorbed as the coefficient matrix of an extra variable.

\subsection{Rational circuits and correctness}

\begin{definition}[circuit, encoding, size]\label{def:circuit}
A \emph{rational circuit} $\Phi$ over $\Q$ in noncommuting variables $x_1,\dots,x_n$ is a directed acyclic graph given as a list of gates $g_1,\dots,g_t$ in topological order, children having smaller index, where each gate is one of the following:
\begin{itemize}[nosep]
\item an input gate labelled by some $x_i$ or by a rational constant $a/b$, with $a,b$ integers in binary and $b\neq 0$ checked syntactically;
\item a fan-in-two gate labelled $+$ or $\times$, with ordered children;
\item a fan-in-one gate labelled $(\,)^{-1}$.
\end{itemize}
One gate is designated the output. The \emph{size} $s$ is the number of gates plus the total bit-length of all constants, so that $n \le s$ and each constant has bit-size at most $s$. We write $\Phi_v$ for the subcircuit rooted at gate $v$, and call a circuit a \emph{formula} when its graph is a tree.
\end{definition}

\begin{definition}[semantics, correctness]\label{def:correct}
Define a partial map $\val$ from gates to $\ff{x_1,\dots,x_n}$ by structural recursion. The value of an input gate is its label; $\val(g+h) = \val(g)+\val(h)$ and $\val(g\times h) = \val(g)\val(h)$, defined if and only if both arguments are; and $\val(g^{-1}) = \val(g)^{-1}$, defined if and only if $\val(g)$ is defined and nonzero. The circuit $\Phi$ is \emph{correct} if $\val$ is defined at every gate, and then $\widehat{\Phi} := \val(\text{output})$.
\end{definition}

\begin{lemma}[integer constants]\label{lem:intconst}
Given a circuit $\Phi$ of size $s$, one can construct in space $O(s)$ a circuit $\Phi'$ of size $s' \le 4s$ whose constants are integers of bit-size at most $s$, such that every gate of $\Phi$ corresponds to a gate of $\Phi'$ computing the same element of the free skew field, and $\Phi'$ is correct if and only if $\Phi$ is. Each inserted inversion gate is legal whenever it is evaluated.
\end{lemma}

\begin{proof}
Replace each constant gate $c = a/b$, with $b \ge 1$ after normalizing signs, by four gates: integer atoms $a$ and $b$, an inversion $b^{-1}$, and a product $a \times b^{-1}$. A nonzero integer $b$ is a nonzero central scalar of the free skew field and hence invertible, so the inserted inversion is always legal, and $ab^{-1} = a/b$. For the size, suppose $\Phi$ has $t$ gates of which $t_0$ are constant gates, and the constants have total bit-length $B$. Then $s = t+B$ and $s' \le (t + 3t_0) + B \le 4t + B \le 4s$.
\end{proof}

From now on we assume that all constants of $\Phi$ are integers of bit-size at most $s$. By Lemma~\ref{lem:intconst} this costs a replacement of $s$ by $s' \le 4s$, which changes none of the bounds below.

\section{From circuits to succinct pencils}\label{sec:pencil}

As stated in the abstract and introduction, Hrubes-Wigderson \cite{HW15} present an algorithm that takes as input a rational formula $\r$ of size $s$, in noncommuting variables $x_1,x_2,\ldots,x_n$ (over a field $F$
that is either $\Q$ or a finite field), and in polynomial-time computes a linear pencil $L$ of dimension $O(s)$ such that $\r$ is correct and
non-zero iff $L$ has full noncommutative rank.

Based on \cite{HW15}, in this section we present an algorithm that takes as input a rational \emph{circuit} $C$ of size $s$ and computes a succinct representation for a linear pencil of dimension $2^{O(s)}$. We remark here that the standard transformation of a rational circuit of size $s$ and depth $d$, with $+$ and $\times$ gates of fanin $2$ and inverse gates of fanin $1$, into a rational formula will yield a formula of size $2^{O(d)}\cdot s$ and depth $d$. In general, as $d\le s$ we obtain a formula of size $2^{O(s)}$. Thus, the \cite{HW15} construction of linear pencil $L$ corresponding to circuit $C$ will have dimension $2^{O(s)}$. More precisely, the $\PSPACE$ algorithm that we will describe takes as input $C,i,j$, where $i$ and $j$ are $O(s)$ bit integers indexing into the linear pencil $L$, and it outputs the entry $L[i,j]$. The algorithm takes space $s^{O(1)}$.

In the described construction we fix the field $F=\Q$.

\subsection{Rational Circuits to Rational Formulas}

The following lemma summarizes the standard transformation of noncommutative rational circuits to rational formulas.

\begin{lemma}\label{lem:unroll}

  For each gate $v$ of $\Phi$ define the formula $F_v$ by structural recursion, so that $F_v$ is the atom of $v$ if $v$ is an input gate, $F_v = F_g \circ F_h$ if $v = g \circ h$ for $\circ \in \{+,\times\}$, and $F_v = (F_g)^{-1}$ if $v = g^{-1}$, with shared subcircuits duplicated. Then:

  \begin{enumerate}[label=(\roman*),nosep]
  \item $\operatorname{depth}(F_v) = \operatorname{depth}(v) \le s-1$, and $F_v$ has at most $2^{s}$ nodes;
  \item the nodes of $F_v$ correspond to the directed paths of $\Phi$ that start at $v$, where a path records at each step which input port of the current gate is followed, and the subformula sitting at a node depends only on the gate at which its path ends;
  \item if $\Phi_v$ is correct then $F_v$ is a correct formula and $\widehat{F_v} = \widehat{\Phi_v}$.
  \end{enumerate}
\end{lemma}

\begin{proof}

  For (i) and (ii), the fan-in is at most $2$, and the vertex sequence underlying such a path is a directed path in a DAG and hence has pairwise distinct vertices. Recording the port is what makes the correspondence a bijection when a binary gate uses the same child at both inputs. For (iii), induction on depth shows that every tree copy of a DAG gate computes the same element of the free skew field as that gate. Correctness of $F_v$ quantifies over all tree copies of the inversion gates of $\Phi_v$, and copies compute identical values.

\end{proof}

\subsection{The Hrube\v{s}--Wigderson pencil and its structure}

In this subsection we recall some details of the Hrube\v{s}--Wigderson pencil construction in a self-contained form that
is suitable for describing our algorithm.

To every (noncommutative rational) formula $F$ we associate a dimension $p(F)$ and a matrix $A_F$ over $\fa{X}$, by the rules of \cite[Thm.~2.6 and \S 6]{HW15}. Below, $u_k = (1,0,\dots,0)$ and $v_k = (0,\dots,0,1)$ are $1\times k$ row vectors, and unspecified entries are $0$.

\begin{definition}[syntactic pencil {\cite{HW15}}]\label{def:pencil}
\begin{align*}
\text{atom } a:\quad & p = 2, & A_F &= \begin{pmatrix} 1 & a\\ 0 & -1\end{pmatrix};\\[4pt]
F = G\times H:\quad & p = p(G)+p(H), & A_F &= \begin{pmatrix} A_G & -v_{p(G)}^{\mathsf{T}} u_{p(H)}\\ 0 & A_H \end{pmatrix};\\[4pt]
F = G+H:\quad & p = p(G)+p(H)+1, & A_F &= \begin{pmatrix} A_G & a^{\mathsf{T}} u_{p(H)} & -v_{p(G)}^{\mathsf{T}}\\ 0 & A_H & v_{p(H)}^{\mathsf{T}}\\ 0 & 0 & 1 \end{pmatrix},\\[-2pt]
 & & &\qquad\text{where } a^{\mathsf{T}} = \text{first column of } A_G;\\[4pt]
F = G^{-1}:\quad & p = p(G)+1, & A_F &= \begin{pmatrix} v_{p(G)}^{\mathsf{T}} & A_G\\ 0 & -u_{p(G)} \end{pmatrix}.
\end{align*}
The construction is purely syntactic, so $A_F$ is defined for every formula, correct or not.
\end{definition}

Spelled out entrywise, which is what the entry oracle of Lemma~\ref{lem:oracle} implements, the three composite rules say the following, with $p = p(G)$ and $q = p(H)$.
\begin{itemize}[nosep]
\item For $F = G\times H$, the block $-v_{p}^{\mathsf{T}}u_{q}$ occupies rows $1,\dots,p$ and columns $p+1,\dots,p+q$, and its only nonzero entry is $-1$, at global position $(p,\,p+1)$.
\item For $F = G+H$, the block $a^{\mathsf{T}}u_{q}$ has first column equal to the first column of $A_G$ and all its other columns zero, and the last column of $A_F$ carries $-1$ at row $p$, $+1$ at row $p+q$, and $+1$ at row $p+q+1$.
\item For $F = G^{-1}$, column $1$ of $A_F$ equals $e_{p}$, the block $A_G$ occupies rows $1,\dots,p$ and columns $2,\dots,p+1$, and the entry at position $(p+1,2)$ is $-1$.
\end{itemize}

\begin{example}\label{ex:pencil}
Take $F = x_1 \times x_2$, a product of two atoms, so that $p(F) = 4$. Then
\[
A_F \;=\;
\begin{pmatrix}
1 & x_1 & 0 & 0\\
0 & -1 & -1 & 0\\
0 & 0 & 1 & x_2\\
0 & 0 & 0 & -1
\end{pmatrix},
\qquad
A_F^{-1} \;=\;
\begin{pmatrix}
1 & x_1 & x_1 & x_1x_2\\
0 & -1 & -1 & -x_2\\
0 & 0 & 1 & x_2\\
0 & 0 & 0 & -1
\end{pmatrix},
\]
and indeed $\bigl(A_F^{-1}\bigr)_{1,4} = x_1x_2 = \widehat{F}$, as Lemma~\ref{lem:invariant} asserts. The wrapper pencil of Theorem~\ref{thm:zerocrit} is the $5\times 5$ matrix
\[
L(F) \;=\; A_{F^{-1}} \;=\;
\begin{pmatrix}
0 & 1 & x_1 & 0 & 0\\
0 & 0 & -1 & -1 & 0\\
0 & 0 & 0 & 1 & x_2\\
1 & 0 & 0 & 0 & -1\\
0 & -1 & 0 & 0 & 0
\end{pmatrix},
\]
which is full because $\widehat{F} \neq 0$.
\end{example}

\begin{lemma}[structure of the pencil]\label{lem:structure}
For every formula $F$, with $p = p(F)$:
\begin{enumerate}[label=(\alph*),nosep]
\item \emph{(first column)} the first column of $A_F$ is the standard basis column $e_{\rho(F)}$, where
\[
\rho(\mathrm{atom}) = 1,\qquad \rho(G\times H) = \rho(G+H) = \rho(G),\qquad \rho(G^{-1}) = p(G),
\]
so in particular it contains only the entries $0$ and $1$, and no variables;
\item \emph{(single symbol)} every entry of $A_F$ is exactly one symbol, namely a scalar in $\{0,\pm 1\}$, a constant of $F$, or a variable with coefficient $1$. Hence
\[
A_F \;=\; Q_0 + \sum_{k=1}^{n} x_k Q_k, \qquad Q_k \in \{0,1\}^{p\times p}\ (k \ge 1),
\]
with $Q_0$ carrying only $0$, $\pm 1$ and constants of $F$;
\item \emph{(occurrences)} each occurrence of $x_k$ in $F$ contributes exactly one entry $1$ to $Q_k$, precisely at the $(1,2)$-position of its $2\times 2$ atom block; distinct occurrences occupy distinct global positions, and every global position is assigned by exactly one clause of Definition~\ref{def:pencil};
\item \emph{(dimensions)} $p(F) \le 3\cdot 2^{\operatorname{depth}(F)} - 1$ and $p(F) \le 2\,|F|$, where $|F|$ is the number of nodes. For $F = F_v$ as in Lemma~\ref{lem:unroll} we get $p(F_v) \le 3\cdot 2^{\,s-1}-1$, which is an $(s+2)$-bit number.
\end{enumerate}
\end{lemma}

\begin{proof}
(a) Induction. An atom's first column is $(1,0)^{\mathsf{T}}$, the product and sum rules copy the first column of $A_G$ padded with zeros, and the first column of the inverse gadget is $e_{p(G)}$.

(b) and (c) Induction. Every position of $A_F$ is one of the following: a structural scalar or zero; an injectively embedded child position; the atom symbol at a $(1,2)$-position; or the copy $A_F[i,\,p(G){+}1] = A_G[i,1]$, which is a $0/1$ scalar by part (a). The last case is the crucial one, since it shows that no variable occurrence is ever duplicated. The structural scalars and zeros are, for the product rule, the entry at $(p(G),p(G){+}1)$ and the off-diagonal zero blocks; for the sum rule, the entries of the last column, the zero columns $p(G){+}2,\dots,p(G){+}p(H)$ restricted to rows $\le p(G)$, and the zero blocks below $A_G$ and in the last row; and for the inverse rule, column $1$ and the last row. The index ranges of the children in every rule are disjoint, so the clause assignment partitions $[p]\times[p]$.

(d) For the depth bound, atoms give $2 = 3\cdot 2^0 - 1$, and a sum of two subformulas of depth at most $d-1$ gives at most $2(3\cdot 2^{d-1}-1)+1 = 3\cdot 2^{d}-1$, while products and inversions give less. For the size bound, $2 = 2\cdot 1$ at atoms, and each rule adds at most $2$ to the dimension while consuming one node.
\end{proof}

\begin{lemma}[invariant for correct formulas; explicit form of {\cite[Thm.~2.6, Prop.~7.1]{HW15}}]\label{lem:invariant}
If $F$ is correct then $A_{F'}$ is invertible over $\ff{X}$ for every subformula $F'$ of $F$, and
\[
\bigl(A_F^{-1}\bigr)_{1,\,p(F)} \;=\; \widehat{F}.
\]
\end{lemma}

\begin{proof}
Induction on $F$. Write $f_G = (A_G^{-1})_{1,p(G)}$ and similarly $f_H$.

\emph{Atom.} $A_F^2 = I$, so $A_F^{-1} = A_F$ and the $(1,2)$ entry is the atom.

\emph{Product.} $A_F$ is block-triangular with invertible diagonal blocks, hence invertible, and
\[
A_F^{-1} \;=\; \begin{pmatrix} A_G^{-1} & A_G^{-1} v_{p(G)}^{\mathsf{T}} u_{p(H)} A_H^{-1}\\ 0 & A_H^{-1}\end{pmatrix}.
\]
Its $(1,p)$ entry is $\bigl(A_G^{-1}v_{p(G)}^{\mathsf{T}}\bigr)_1\bigl(u_{p(H)}A_H^{-1}\bigr)_{p(H)} = f_G\, f_H$, with the factors in this order.

\emph{Sum.} $A_F$ is block-triangular with invertible diagonal, hence invertible. Solving $A_F w = e_{p}$ blockwise, the top block of $w$ is
\[
A_G^{-1} v_{p(G)}^{\mathsf{T}} \;+\; A_G^{-1} a^{\mathsf{T}}\,\bigl(u_{p(H)} A_H^{-1} v_{p(H)}^{\mathsf{T}}\bigr).
\]
Now $u_{p(H)}A_H^{-1}v_{p(H)}^{\mathsf{T}} = f_H$ is a single element acting on the right, and $A_G^{-1}a^{\mathsf{T}} = e_1$ because $a^{\mathsf{T}} = A_G e_1$. So the first entry is $f_G + f_H$.

\emph{Inverse.} Correctness gives $G$ correct and $b := \widehat{G} \neq 0$, and by induction $A_G$ is invertible with $f_G = b$. With $w = -A_G^{-1}v_{p(G)}^{\mathsf{T}}$ and $y = u_{p(G)}A_G^{-1}$ we have
\[
\begin{pmatrix} I_{p(G)} & 0\\ y & 1\end{pmatrix}\, A_F\, \begin{pmatrix} 1 & 0\\ w & I_{p(G)}\end{pmatrix} \;=\; \begin{pmatrix} 0 & A_G\\ b & 0\end{pmatrix},
\]
which is invertible since $b \neq 0$. Unwinding, $\bigl(A_F^{-1}\bigr)_{1,p(G)+1} = b^{-1} = \widehat{F}$.
\end{proof}

\begin{lemma}[inversion gadget over a skew field]\label{lem:gadget}
Let $D$ be a skew field, let $A \in D^{p\times p}$, and let
\[
B \;=\; \begin{pmatrix} v_p^{\mathsf{T}} & A\\ 0 & -u_p\end{pmatrix} \;\in\; D^{(p+1)\times(p+1)}.
\]
If $A$ is invertible then $B$ is invertible if and only if $(A^{-1})_{1,p}\neq 0$.
\end{lemma}

\begin{proof}
This is the two-sided reduction displayed in the inverse case of Lemma~\ref{lem:invariant}, which is valid whenever $A^{-1}$ exists: $B$ is equivalent to $\operatorname{diag}(A,\,b)$ with $b = (A^{-1})_{1,p}$.
\end{proof}

\begin{remark}\label{rem:notsharper}
The unconditional biconditional ``$B$ is full if and only if $A$ is full and $(A^{-1})_{1,p}\neq 0$'' is false. For $A = \begin{psmallmatrix}0&1\\0&0\end{psmallmatrix}$, which is not full, $B$ is a signed permutation matrix and hence full. Only the conditional form above holds. We apply it exclusively to pencils $A_F$ of correct formulas $F$, which are invertible by Lemma~\ref{lem:invariant}.
\end{remark}

\begin{theorem}[zero criterion; after {\cite[\S 7]{HW15}}]\label{thm:zerocrit}
Let $F$ be a correct formula, let $p = p(F)$, and let
\[
L(F) \;:=\; A_{F^{-1}} \;=\; \begin{pmatrix} v_p^{\mathsf{T}} & A_F\\ 0 & -u_p \end{pmatrix}, \qquad N := p+1 .
\]
This is a syntactic construction, defined whether or not $F^{-1}$ is correct. The following are equivalent:
\emph{(1)} $\widehat{F} = 0$;\quad
\emph{(2)} $F^{-1}$ is not a correct formula;\quad
\emph{(3)} $L(F)$ is not invertible over $\ff{X}$;\quad
\emph{(4)} $L(F)$ is not full.
Moreover $L(F) = Q_0 + \sum_{k=1}^{n} x_k Q_k$ with $Q_k \in \{0,1\}^{N\times N}$ for $k \ge 1$, carrying one entry $1$ per occurrence of $x_k$ in $F$, and with $Q_0$ carrying only $0$, $\pm 1$, and constants of $\Phi$, which are integers of bit-size at most $s$ after Lemma~\ref{lem:intconst}.
\end{theorem}

\begin{proof}
For (1) $\Leftrightarrow$ (2), note that $F^{-1}$ appends one inversion gate to the correct formula $F$, so $F^{-1}$ is correct if and only if $\widehat{F}\neq 0$. For (1) $\Leftrightarrow$ (3), by Lemma~\ref{lem:invariant} $A_F$ is invertible with $(A_F^{-1})_{1,p} = \widehat{F}$, and by Lemma~\ref{lem:gadget} $L(F)$ is invertible if and only if $\widehat{F} \neq 0$. For (3) $\Leftrightarrow$ (4), apply Proposition~\ref{prop:dictionary}. The structure statement is Lemma~\ref{lem:structure} applied to $F^{-1}$.
\end{proof}

\subsection{The entry oracle}

\begin{lemma}[entry oracle]\label{lem:oracle}
There is a deterministic algorithm which takes as input a rational circuit $\Phi$ of size $s$ with integer constants, a gate $v$, 
and a query $(k,i,j)$ with $0 \le k \le n$ and $1 \le i,j \le N$, where $N = p(F_v)+1$ is the dimension of the pencil
\[
L_v \;:=\; A_{(F_v)^{-1}} \;=\; Q_0 + \sum_{k} x_k Q_k ,
\]
and outputs the entry $(Q_k)_{i,j}$ in workspace $O(s^{O(1)})$. The output is $0$, $\pm 1$, or a single integer constant 
of $\Phi$ of bit-size at most $s$.
\end{lemma}

\begin{proof}
\emph{Preprocessing.} Compute $p(g)$ for every gate $g$ in topological order via the recurrences of Definition~\ref{def:pencil}:
\[
p(\mathrm{atom}) = 2,\qquad p(g\times h) = p(g)+p(h),\qquad p(g+h) = p(g)+p(h)+1,\qquad p(g^{-1}) = p(g)+1 .
\]
By Lemma~\ref{lem:unroll}(ii) the unrolled subformula at a tree node depends only on the gate at which its path ends, so $p(g) = p(F_g)$ is well defined. By Lemma~\ref{lem:structure}(d) each value fits in $s+2$ bits, so the array of values $p(g)$ requires $O(s^{2})$ bits and is computed with $O(s)$-bit additions.

\emph{Descent.} Maintain a state $(g,i,j)$, initialized at the virtual gate $v' := v^{-1}$ with $p(v') = p(v)+1$, whose gadget is the inverse gadget. Repeat the unique matching clause below, writing $P = p(g)$ and, where applicable, $p = p(g_L)$ and $q = p(g_R)$ for the children. Each clause either outputs a value and halts, or replaces the state and continues. An output clause returning a scalar $a$ answers the original query as $a$ if $k = 0$ and as $0$ if $k \ge 1$, and the atom clause for a variable $x_\ell$ answers $1$ if $k = \ell$ and $0$ otherwise.

For an atom with symbol $a$, so that $P = 2$:
\[
(1,1)\mapsto 1,\qquad (1,2)\mapsto a,\qquad (2,1)\mapsto 0,\qquad (2,2)\mapsto -1 .
\]
For $g = g_L \times g_R$, so that $P = p+q$:
\[
(i,j) \;\longmapsto\;
\begin{cases}
(g_L,\;i,\;j), & i\le p,\ j \le p,\\[2pt]
(g_R,\;i-p,\;j-p), & i > p,\ j > p,\\[2pt]
\text{output } -1, & (i,j) = (p,\,p+1),\\[2pt]
\text{output } 0, & \text{otherwise.}
\end{cases}
\]
For $g = g_L + g_R$, so that $P = p+q+1$:
\[
(i,j) \;\longmapsto\;
\begin{cases}
(g_L,\;i,\;j), & i\le p,\ j\le p,\\[2pt]
(g_R,\;i-p,\;j-p), & p < i \le p+q,\ p< j \le p+q,\\[2pt]
(g_L,\;i,\;1), & i \le p,\ j = p+1,\\[2pt]
\text{output } -1, & j = P,\ i = p,\\[2pt]
\text{output } +1, & j = P,\ i \in \{p+q,\;P\},\\[2pt]
\text{output } 0, & \text{otherwise.}
\end{cases}
\]
For $g = g_L^{-1}$, so that $P = p+1$:
\[
(i,j) \;\longmapsto\;
\begin{cases}
\text{output } 1, & j = 1,\ i = p,\\[2pt]
(g_L,\;i,\;j-1), & i \le p,\ j \ge 2,\\[2pt]
\text{output } -1, & i = P,\ j = 2,\\[2pt]
\text{output } 0, & \text{otherwise.}
\end{cases}
\]

\emph{Correctness.} Each clause list is exactly the entrywise description following Definition~\ref{def:pencil}. The clauses partition $[P]\times[P]$, the ranges being disjoint and exhaustive because $p,q \ge 2$, and each clause has at most one continuation. In particular the clause $(i,\,p{+}1)\mapsto(g_L,\,i,\,1)$ delegates to the first column of the child, which is a $0/1$ scalar by Lemma~\ref{lem:structure}(a). Hence the invariant
\begin{quote}
\emph{the answer to the original query is the coefficient of $x_k$, or the scalar part if $k=0$, of the entry of the current pencil at the current position $(i,j)$}
\end{quote}
is preserved, and the value output at the bottom is $(Q_k)_{i,j}$ by Lemma~\ref{lem:structure}(b) and (c).

\emph{Termination and resources.} Every continuation of the algorithm moves to a child gate, one level down the unrolled tree. The gates visited form a directed path of the DAG, so they are pairwise distinct and there are at most $s+1$ steps. The state consists of a gate identifier of $O(\log s)$ bits, the indices $i,j$ of $s+2$ bits each, the query register, and $O(s)$ scratch bits. No recursion stack is required, since every step replaces the state. The total workspace is $O(s^{O(1)})$ for the table plus the space for the continuations.
\end{proof}

\section{Homogenization}\label{sec:homog}

\begin{lemma}[homogenization]\label{lem:homog}
Let $L = Q_0 + \sum_{k=1}^{n} Q_k x_k$ be an affine pencil of dimension $N$ over $\Q$ with $n \ge 1$, and let
\[
L^{h} \;:=\; Q_0\,x_0 + \sum_{k=1}^{n} Q_k\,x_k
\]
over the variables $x_0, x_1, \dots, x_n$. Then $\ncrk(L) = \ncrk(L^{h})$. In particular $L$ is full over $\ff{x_1,\dots,x_n}$ if and only if $L^h$ is full over $\ff{x_0,\dots,x_n}$.
\end{lemma}

\begin{proof}
$\ncrk(L)\le\ncrk(L^h)$ is clear. To see that $\ncrk(L)\geq \ncrk(L^h)$, we note that for random matrix substitution $(X_0, \ldots, X_n)$ of suitably large dimension $d$ the rank of scalar matrix  $Q_0 \otimes X_0 + \sum_{k=1}^{n} Q_k \otimes X_k$ is $r \cdot d$ where $r = \ncrk(L^h)$. Hence, rank of the scalar matrix $Q_0 \otimes I_d + \sum_{k=1}^{n} Q_k \otimes X_k X_0 ^{-1}$ is equal to $r \cdot d$ implying $\ncrk(L)\geq \ncrk(L^h)$.
\end{proof}

\begin{definition}[the problem $\ncFULL$]\label{def:ncfull}
An instance of $\ncFULL$ is a tuple $(N,m,w,A_1,\dots,A_m)$ with $A_i \in \Z^{N\times N}$, encoded in binary as follows: a self-delimiting header specifying $N$, $m$ and $w$ in binary, followed by the $mN^{2}$ entries in row-major order, each written in two's complement in exactly $w$ bits. The instance is a \textsc{yes}-instance if the homogeneous pencil $\sum_{i=1}^{m}A_ix_i$ has full noncommutative rank over $\Q$, that is, if it is full in the sense of Definition~\ref{def:ncrank}.
\end{definition}

\begin{theorem}[ {\cite[Thm.~1.5 and Cor.~5.6]{CGGRT26}}]\label{thm:cggrt}
$\ncFULL \in \NC^{2}$. That is, $\ncFULL$ is decided by a uniform family of Boolean circuits, in the sense of Definition~\ref{def:uniform}, of size $M^{O(1)}$, depth $O(\log^{2}M)$ and uniformity space $O(\log M)$, where $M$ is the bit-length of the encoded instance.
\end{theorem}

\section{Rational identity testing is in polynomial space}\label{sec:main}

\begin{proposition}\label{prop:driver}
Let $\mathrm{ZERO}(\cdot)$ be any procedure which, on every correct rational circuit $\Psi$ of size at most $s$, decides whether $\widehat{\Psi} = 0$ in workspace $Z(s)$. Consider the algorithm \textup{\textsc{Check}}: let $u_1,\dots,u_t$ with $t \le s$ be the inversion gates of $\Phi$ in the input's topological order, with children $v_1,\dots,v_t$; for $i = 1,\dots,t$ run $\mathrm{ZERO}(\Phi_{v_i})$, and if the answer is ``zero'', output ``$\Phi$ is not correct: gate $u_i$ inverts $0$'' and halt; if all $t$ tests pass, output ``$\Phi$ is correct'' together with the answer of $\mathrm{ZERO}(\Phi)$. Then:
\begin{enumerate}[label=(\alph*),nosep]
\item every invocation of $\mathrm{ZERO}$ satisfies its promise, that is, its argument is a correct circuit;
\item \textup{\textsc{Check}} outputs ``not correct'' if and only if $\Phi$ is incorrect, and the reported gate $u_i$ genuinely inverts the zero element, which is a well-defined value;
\item if $\Phi$ is correct, \textup{\textsc{Check}} decides whether $\widehat{\Phi} = 0$.
\end{enumerate}
\textup{\textsc{Check}} makes at most $t+1 \le s+1$ calls, sequentially and with space reuse, in workspace $Z(s) + O(s)$.
\end{proposition}

\begin{proof}
For $1 \le i \le t+1$ let $C_i$ be the assertion:
\begin{quote}
\emph{if tests $1,\dots,i-1$ all returned ``nonzero'', then every inversion gate of $\Phi$ preceding $u_i$, and all of them if $i = t+1$, has a defined and nonzero argument.}
\end{quote}
We prove $C_i$ by strong induction. Consider $u_j$ with $j < i$ and its child $v_j$. Every inversion gate inside $\Phi_{v_j}$ is reachable from $v_j$ and hence precedes $u_j$ in every topological order, so it equals some $u_{j'}$ with $j' < j$. By induction, using $C_j$, all of them had defined and nonzero arguments when test $j$ ran. So $\Phi_{v_j}$ was correct, since induction along a topological order of $\Phi_{v_j}$ shows that $\val$ is defined at all its gates, the call met its promise, and its truthful answer ``nonzero'' means $\val(v_j) \neq 0$. This establishes $C_i$, and also (a).

For (b), if test $i$ is the first to return ``zero'' then $\Phi_{v_i}$ is correct and $\val(v_i) = 0$ is a well-defined zero, minimality ruling out ``undefined'', so $u_i$ inverts $0$. Conversely, if $\Phi$ is incorrect then the tests cannot all pass, by $C_{t+1}$. Part (c) is immediate from (a) and $C_{t+1}$.
\end{proof}

\begin{theorem}\label{thm:main}
There is a deterministic Turing machine which, given a rational circuit $\Phi$ of size $s$ over $\Q$ in noncommuting variables $x_1,\dots,x_n$, decides in workspace $O(s^{O(1)})$:
\begin{enumerate}[label=(\alph*),nosep]
\item whether $\Phi$ is a correct circuit, identifying an inversion gate whose argument evaluates to $0$ if it is not; and
\item if $\Phi$ is correct, whether $\widehat{\Phi} = 0$ in $\ff{x_1,\dots,x_n}$.
\end{enumerate}
In particular RIT for noncommutative circuits over $\Q$, with no correctness promise, is in $\PSPACE$. 

\end{theorem}

\begin{proof}
Given the rational circuit $\Phi$ as input, by Lemma~\ref{lem:oracle} we first obtain the succinct representation of its linear pencil $L$ of size $2^{O(s)}$,
where the succinct representation is described by a deterministic $s^{O(1)}$ space algorithm that takes as input the circuit $\Phi$ and indices $k$, $i$ and
$j$ and outputs the entry $Q_k([i,j]$ of the linear pencil $L=Q_0+\sum_{k=1}^n x_kQ_k$. Now, we will simulate the $\NC$ algorithm of 
Theorem~\ref{thm:cggrt} on the succinct input $L$ using the standard polylogspace simulation of $\NC$ circuits (Corollary~\ref{cor:succNC}), where the
space bound is polylogarithmic in the size of the input $L$ (which is $2^{O(s)}$). This amounts to simulating a space $s^{O(1)}$
algorithm on the succinctly presented input $L$ (where the succinct presentation is also given by an $s^{O(1)}$ space algorithm). 
Putting these together immediately gives an $s^{O(1)}$ space bounded algorithm for checking if $\Phi$ is correct and, if so, whether
or not it is identically zero.  
\end{proof}

The following Corollary is immediate.

\begin{corollary}[division-free circuits]\label{cor:pit}
Polynomial identity testing of noncommutative circuits (without degree restriction) over $\Q$, computing polynomials of degree up to $2^{s}$, 
is in $\PSPACE$.
\end{corollary}

\begin{remark}
Theorem \ref{thm:main} and Corollary \ref{cor:pit} hold for rational circuits over any finite field without any changes in the proofs.
\end{remark}

\subsection*{Acknowledgements}


In preparing this paper we used Anthropic's Claude (Opus~5). The model was used to draft and structure the exposition from our outline. The 
results and text are the authors' responsibility.

\end{document}